\documentclass[preprint,superscriptaddress,compsoc,nofootinbib]{revtex4-1}

\usepackage{algorithm}

\usepackage{xcolor}

\usepackage{graphicx}
\usepackage{lipsum}
\usepackage{mathtools}
\usepackage{amsmath,amssymb,amsthm}

\usepackage{lineno}
\newcommand*\patchAmsMathEnvironmentForLineno[1]{%
  \expandafter\let\csname old#1\expandafter\endcsname\csname #1\endcsname
  \expandafter\let\csname oldend#1\expandafter\endcsname\csname end#1\endcsname
  \renewenvironment{#1}%
     {\linenomath\csname old#1\endcsname}%
     {\csname oldend#1\endcsname\endlinenomath}}%
\newcommand*\patchBothAmsMathEnvironmentsForLineno[1]{%
  \patchAmsMathEnvironmentForLineno{#1}%
  \patchAmsMathEnvironmentForLineno{#1*}}%
\AtBeginDocument{%
\patchBothAmsMathEnvironmentsForLineno{equation}%
\patchBothAmsMathEnvironmentsForLineno{align}%
\patchBothAmsMathEnvironmentsForLineno{flalign}%
\patchBothAmsMathEnvironmentsForLineno{alignat}%
\patchBothAmsMathEnvironmentsForLineno{gather}%
\patchBothAmsMathEnvironmentsForLineno{multline}%
}

\theoremstyle{plain}
\newtheorem{lemma}{Lemma}

\newtheorem{theorem}{Theorem}
\newtheorem{proposition}{Proposition}
\theoremstyle{definition}
\newtheorem{definition}{Definition}

\newcommand{\prefs}{\mathcal{P}}
\newcommand{\freq}[1]{\mu_\pi(\succ_{#1})}
\newcommand{\dist}[1]{d(\succ_{#1},\succ_0)}
\newcommand{\switch}{w^{ab}}
\newcommand{\pconfixed}{P^{\text{con-fixed}}}
\newcommand{\pcon}{P^{\text{consensus}}}
\newcommand{\pfcon}{P^{\text{flexible-consensus}}}

\newcommand{\peq}{P^{\text{equal}}}

\begin{document}

\title{Flexible Level-1 Consensus Ensuring Stable Social Choice: Analysis and Algorithms}

\author{Mor Nitzan}

\affiliation{School of Computer Science, The Hebrew University, 
Jerusalem 91904, Israel}

\affiliation{Racah Institute of Physics, The Hebrew University, 
Jerusalem 91904, Israel}

\affiliation{Department of Microbiology and Molecular Genetics, 
Faculty of Medicine, The Hebrew University,
Jerusalem 91120, Israel}

\author{Shmuel Nitzan} 

\affiliation{Department of Economics, Bar-Ilan University, Ramat Gan 5290002, Israel}

\author{Erel Segal-Halevi} 

\affiliation{Department of Computer Science, Ariel University, Ariel 40700, Israel}

\affiliation{Department of Computer Science, Bar-Ilan University, Ramat Gan 5290002, Israel}


\begin{abstract}
\emph{Level-1 Consensus} is a property of a preference-profile. Intuitively, it means that there exists a preference relation which induces an ordering of all other preferences such that frequent preferences are those that are more similar to it. This is a desirable property, since it enhances the stability of social choice by guaranteeing that there exists a Condorcet winner and it is elected by all scoring rules.

In this paper, we present an algorithm for checking whether a given preference profile exhibits level-1 consensus. We apply this algorithm to a large number of preference profiles, both real and randomly-generated, 
and find that level-1 consensus is very improbable. We support 
these empirical findings theoretically, by showing that, under the impartial culture assumption, the probability of level-1 consensus approaches zero when the number of individuals approaches infinity.

Motivated by these observations, we show that the level-1 consensus property can be weakened while retaining its stability implications. We call this weaker property \emph{Flexible Consensus}. We show, both empirically and theoretically, that it is considerably more probable than the original level-1 consensus. In particular, under the impartial culture assumption, the probability for Flexible Consensus converges to a positive number when the number of individuals approaches infinity.

JEL classification number: D71

\textbf{Keywords}: domain restrictions, level-1 consensus, single peakedness.
\end{abstract}
\maketitle
\section{Introduction}
Recently, \citet{mahajne2015} have proposed the concept of \emph{level-1 consensus} of a preference profile, showing that it considerably enhances the stability of social choice. In particular, if a profile exhibits level-$1$ consensus around a given preference-relation $\succ_0$ with respect to the inversion metric, then:\footnote{In fact, \citet{mahajne2015} define a family of  conditions called level-$r$ consensus, where $r$ is an integer
between $1$ and $K!/2$ and $K$ is the number of alternatives%
. 
But for the sake of simplicity, in the present paper we focus on level-$1$ consensus which is the strongest condition in this family.

Note that recently \citet{poliakov2016note} proved that level-$r$ consensus is equivalent to level-$1$ consensus whenever $r\leq (K-1)!$. 
}
\begin{itemize}
\item There exists a Condorcet winner;
\item The Condorcet winner is chosen by all scoring rules;
\item With an odd number of individuals, the majority relation is transitive and coincides with $\succ_0$.
\end{itemize}
The current study starts by considering two questions: 
\begin{enumerate}
\item How can a profile be tested for level-1 consensus?
\item How likely is it that level-1 consensus exists?
\end{enumerate}
Questions of the former type have been recently studied with respect to various domain restrictions. For example,
\citet{escoffier2008single} provide an efficient way to check whether a profile is \emph{single-peaked},
\citet{bredereck2013} provide an efficient way to check whether a profile is \emph{single-crossing},
and
\citet{barbera2011} ask whether the satisfaction of their proposed \emph{top monotonicity} condition (a sufficient condition for an extension of the median-voter theorem to hold) is easy to check. See \cite{elkind2016preference} for a survey of recent developments in algorithmic checking of domain restrictions.

Questions of the latter type have been studied in the social choice literature with respect to various domain restrictions   that guarantee social stability, e.g., the existence of Condorcet winners under the majority rule \cite{Gehrlein1981Expected,tsetlin2003impartial} and the likelihood of single-peaked preferences \cite{lackner2017likelihood}. 


Our answer to the first question is an efficient algorithm for determining whether a preference-profile exhibits level-1 consensus. In case such a consensus exists, the algorithm identifies a preference relation around which it  occurs.

Our answer to the second question is that level-1 consensus is highly improbable. We applied our algorithm on a recently-released dataset of 315 real-world profiles from various sources \citep{MaWa13a} and found that none of them exhibits level-1-consensus. Moreover, experiments performed  on thousands of profiles generated randomly according to Mallows' phi model \citep{mallows1957} revealed that, for a wide range of parameter settings, profiles exhibiting level-1 consensus were extremely rare. As a partial explanation to these findings, we prove that under the standard probabilistic setting of equally-probable preference relations, the probability of level-1 consensus goes to zero when the number of individuals goes to infinity.

Motivated by these results, we found a way to weaken the level-1 consensus property while keeping its stability properties. We call the weakened property \emph{Flexible Consensus}. In the above mentioned dataset, we found that 39 out of 315 profiles exhibit Flexible Consensus. Flexible Consensus is also much more probable in the settings of the randomly-generated profiles we tested.
In particular, under the impartial-culture assumption, the probability of Flexible Consensus is lower-bounded by a positive constant even when the number of individuals goes to infinity.


\section{Definitions}
Let  $A=\{a_1,\dots,a_K \}$ be a set of $K\geq 3$ alternatives
and $N=\{1,\dots,n\}$ a set of voters.
We assume that each voter has a \emph{strict total order} on the alternatives, i.e, for each two different alternatives $a$ and $b$, either the voter strictly prefers $a$ to $b$ or the voter strictly prefers $b$ to $a$. 
Let 
$\prefs{}$ be the subset of strict total orders on $A$.
We will refer to the elements of $\prefs{}$ as \emph{preference relations} or simply as \emph{preferences}. 

A \emph{preference profile} or simply a \emph{profile} is a list $\pi=(\succ_1,\dots, \succ_n)$ of preference relations on $A$. For each $i \in N$, $\succ_i$ is the preference relation of voter $i$. We denote by $\prefs{}^n$ the set of all possible profiles.

Let $\pi=(\succ_1,\dots, \succ_n)$ be a profile. For each preference $\succ \in \prefs{}$, let
$\mu_\pi(\succ):=|\{i \in N: \succ_i = \succ\}| = $
the number of voters whose preference is $\succ$, 
which in this paper is referred to as the \emph{frequency} of $\succ$.

\begin{definition}
The \emph{inversion distance} between two preferences $\succ,\succ'$, denoted $d(\succ,\succ')$, is the number of pairs of alternatives that are ranked differently by the two preferences, i.e, the number of sets $\{a,b\}\subseteq  A$ such that $a\succ b$ and $b\succ' a$ or vice-versa.
\end{definition}
The inversion-distance is a metric on $\prefs$ \cite{kemeny1962mathematical}. It can vary between $0$ and ${K\choose 2}$, the number of subsets of two alternatives.

For example, if there are three alternatives and $a_1\succ a_3 \succ a_2$ and $a_2 \succ' a_3 \succ' a_1$, then $d(\succ,\succ')=3$ since all three pairs of alternatives are ranked differently by $\succ$ and $\succ'$.

The following definition is due to \citet{mahajne2015}.
\begin{definition}
\label{def:consensus}
Let $\succ_0 \in \prefs$.
A profile $\pi \in \prefs^n$ exhibits \emph{consensus of level-$1$ around $\succ_0$} if the following two conditions hold:
\begin{enumerate}
\item
For all pairs of preferences $\succ, \succ' \in \prefs{}$, $d(\succ,\succ_0) \leq  d(\succ',\succ_0)$ implies $ \mu_\pi(\succ) \ \geq \ \mu_\pi(\succ ')$.
\item 
There is at least one pair $\succ ,\succ' \in \prefs{}$, such that $d(\succ,\succ_0) < d(\succ',\succ_0)$ and $\mu_\pi(\succ) \ > \ \mu_\pi(\succ ')$.
\end{enumerate}
\end{definition}

\section{Algorithm for Detecting Level-1 Consensus}
\label{sec:algorithm-1}

\begin{figure}
\begin{algorithm}[H]
\begin{flushleft}
INPUT: 
\begin{itemize}
\item A set of $K$ alternatives, $A = \{a_1,\ldots,a_K\}$.
\item A profile $\pi$ containing $n$ preference-relations on $A$ (possibly with duplicates).
\end{itemize}

OUTPUT:
\begin{itemize}
\item If $\pi$ exhibits level-1 consensus around some preference $\succ_0$, return  $\succ_0$. 
\item Otherwise, return ``no consensus''.
\end{itemize}

ALGORITHM:
\begin{enumerate}
\item Calculate the frequency $\mu_\pi(\succ)$ of each preference $\succ\in\pi$.
\item Define $n'=$ number of distinct preferences in $\pi$.
\item If all frequencies are equal and $n'=K!$, return ``no consensus; Condition~2 violated''.
\item Order the preferences by descending frequency:
$\mu_\pi(\succ_1) ~\geq~ \mu_\pi(\succ_2) ~\geq~ \cdots~\geq~ \mu_\pi(\succ_{n'})$.
\item Set $M := \freq{1} $ to be the maximum frequency of a preference relation.
\item For $j = 1, 2, \ldots n'$ \  while $\mu_\pi(\succ_j) = M$: \\
Check if Condition~1 is satisfied for $\pi$ and  $\succ_j$ (see Algorithm \ref{alg:condition1}). If it is, then return $\succ_j$.
\item Return ``no consensus; Condition~1 is violated for all candidates''.
\end{enumerate}
\end{flushleft}
\caption{
\label{alg:level1consensus}
Detection of Level-1 Consensus 
}
\end{algorithm}
\vskip -8mm
\end{figure}

\begin{figure}
\begin{algorithm}[H]
\begin{flushleft}
INPUT: 
\begin{itemize}
\item A profile $\pi$ containing $n$ preference-relations on $A$ (possibly with duplicates).
\item A preference $\succ_0$.
\end{itemize}

OUTPUT: ``True'' if Condition~1 is satisfied for profile $\pi$ with respect to $\succ_0$. ``False'' otherwise.

ALGORITHM:

\begin{enumerate}
\item Within every group of preferences with the same frequency, order the preferences by ascending inversion-distance from $\succ_0$.
\item For $i = 1, 2, \ldots, n'-1$:
\\
If $\freq{i} ~ > ~ \freq{i+1}$  and  $\dist{i} \geq \dist{i+1}$,
 return False.
\item Set $\widehat{d}:=\dist{n'}$.
\hskip 1cm
If $\sum_{j=0}^{\widehat{d}} |T(K,j)| = n$, return True. 
\hskip 1cm
Else, return False.
\end{enumerate}
\end{flushleft}
\caption{
\label{alg:condition1}
Check if Condition~1 is satisfied
}
\end{algorithm}
\end{figure}

Given a profile $\pi\in \prefs{}^n$, 
we would like to check whether there exists some preference-relation $\succ_0$ such that $\pi$ exhibits level-1 consensus around it. 
Our solution is given by Algorithm \ref{alg:level1consensus} in page \pageref{alg:level1consensus}
(using, as a sub-routine,  Algorithm \ref{alg:condition1} in page \pageref{alg:condition1}).
\begin{theorem}
\label{thm:level1consensus}
Algorithm \ref{alg:level1consensus} checks whether a profile $\pi$ exhibits level-1 consensus in time:
\begin{align*}
O(n'^2 K\sqrt{\log K} + n'^2 \log{n'})
\end{align*}
where $n'=\min(K!,n)$.
In particular, the run-time is polynomial in the profile size.
\end{theorem}
This entire section is devoted to proving Theorem \ref{thm:level1consensus}. We first explain why Algorithm \ref{alg:level1consensus} is correct. Then we analyze its run-time.

Algorithm \ref{alg:level1consensus} is based on the simple observation that the two conditions in Definition \ref{def:consensus} are equivalent to the following:
\begin{quote}
(Condition~1)
~~
For all $\succ, \succ' \in \prefs{}$,~~if~~$\mu_\pi(\succ') >  \mu_\pi(\succ)$,~
then~$d(\succ', \succ_0)  <  d(\succ, \succ_0)$.
\end{quote}

\begin{quote}
(Condition~2)
~~~
There exists a pair $\succ, \succ' \in \prefs{}$ such that $\mu_\pi(\succ') \ > \ \mu_\pi(\succ)$.
\end{quote}

The algorithm proceeds in several steps.

First, we calculate the frequency $\mu_\pi(\succ)$ of each of the  preferences $\succ \in \pi$. Let $n'$ be the number of distinct preferences in $\pi$. Note that $n'\leq n$ and also $n'\leq K!$, since with $K$ alternatives there are at most $K!$ possible preferences. 
Now Condition~2 is easily checked: it is satisfied if-and-only-if (a) there exists a pair of preferences in $\pi$ with different frequencies, or (b) $n'<K!$ (since this implies that there exists a preference not in $\pi$ with frequency 0).

If Condition~2 is satisfied, it only remains to check whether Condition~1 is satisfied as well.

We order the preferences in descending order of $\mu_\pi(\succ)$, and rename them $\succ_1, \succ_2, \ldots, \succ_{n'}$, such that $\freq{1} \geq \freq{2} \geq \cdots\geq \freq{n'}$. This enables us to identify the candidates for level-1-consensus. Since $d(\succ_0,\succ_0)=0$, Condition~1 immediately implies that each candidate for level-1 consensus must be a preference with maximal frequency. 
So, the candidates are  $\succ_1, \succ_2, \ldots, \succ_{h}$ such that $h \leq n'$ is the largest index for which $\freq{1} = \freq{2} = \cdots = \freq{h}$.

Now we can directly check Condition~1. This condition should be checked separately for each candidate preference $\succ_0$.
Given a candidate-preference $\succ_0$, we can calculate its inversion-distance from each preference $\succ_i \in \pi$ , $\dist{i}$.
Now, we represent the profile $\pi$ relative to $\succ_0$ in the form of a scatter-plot, which will lead to a straight-forward assessment of the profile's consensus status.
Our scatter-plot is a plot whose x-axis denotes the distance $\dist{i}$ and whose y-axis denotes the  frequency $\freq{i}$.
Note that there may be several different preference relations $\succ_i$ with the same frequency,
$\freq{i}=m$. Therefore, for each integer value on the y-axis of the scatter-plot, $m$, we may have several corresponding values on the x-axis, which can be represented by a horizontal segment whose maximum and minimum borders are given by $\max_{i:\freq{i}=m} ~ \dist{i}$ and $\min_{i:\freq{i}=m} ~ \dist{i}$, respectively.

Condition~1 above requires that, for every two frequencies $m_1>m_2$, \emph{all} preferences with frequency $m_1$ are closer to $\succ_0$ than \emph{all} preferences with frequency $m_2$:
$\max_{i:\freq{i}=m_1} ~ \dist{i} < \min_{i:\freq{i}=m_2} ~ \dist{i}$.
Graphically (see Figure \ref{fig:scatters}), this means that when we scan the scatter plot from top to bottom, we must see non-overlapping intervals ordered strictly from left to right.

Three examples are shown in Figure \ref{fig:scatters}. The left example is positive: there are 5 non-overlapping intervals (two of which consist of a single point), and when they are scanned from top to bottom, they are ordered strictly from left to right. Therefore Condition~1 holds. The middle and right examples are negative: the second and third intervals from the top overlap. For instance, in the middle example the overlap is in a single point, $x=2$. This point corresponds to two distinct preferences with different frequencies (5 and 4), both of which are at distance 2 from the candidate preference; these preferences violate Condition~1. 

\begin{figure}
\begin{center}
\includegraphics[width=.3\textwidth]{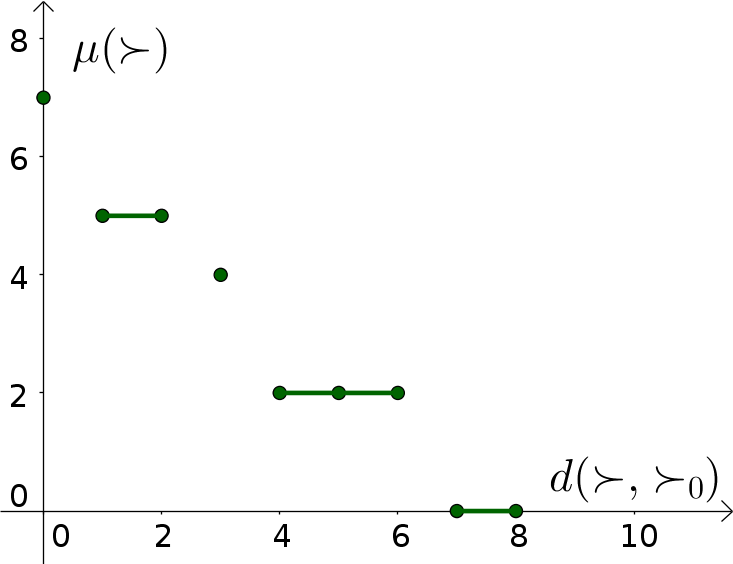}
\hskip .5cm
\includegraphics[width=.3\textwidth]{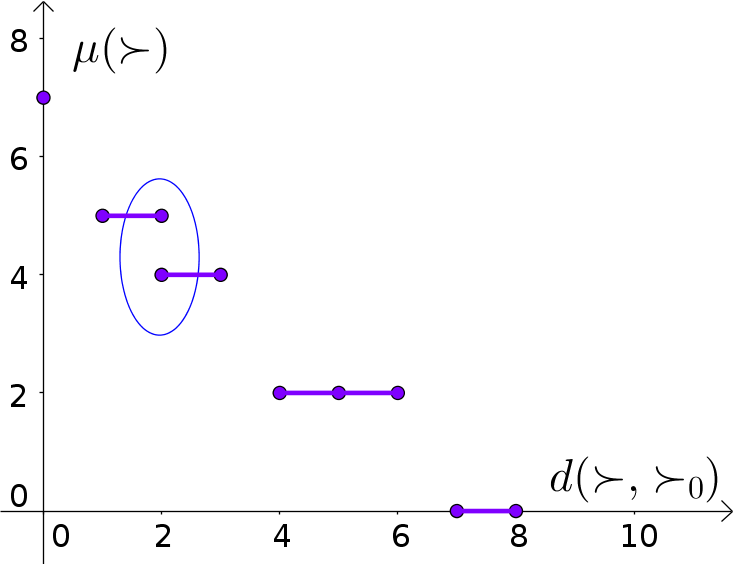}
\hskip .5cm
\includegraphics[width=.3\textwidth]{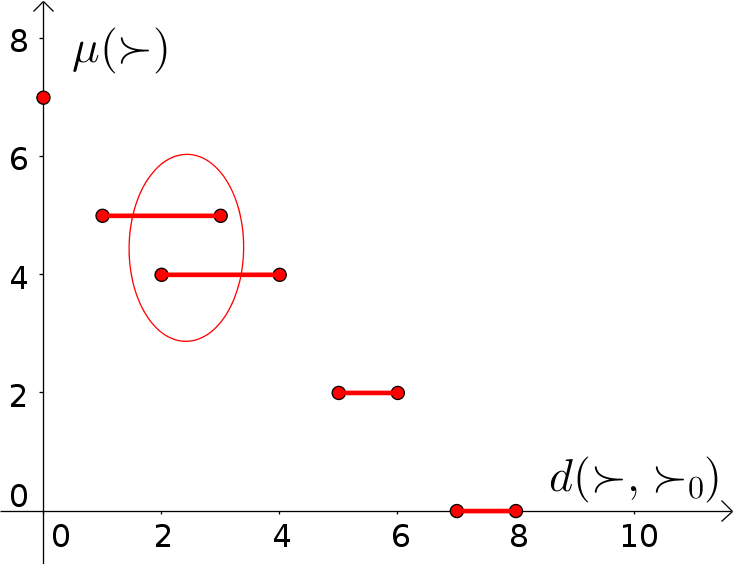}
\end{center}
\caption{\label{fig:scatters} Scatter-plot representation of profiles relative to a candidate preference enables to graphically check whether level-1 consensus around that candidate is satisfied. The left plot satisfies Condition~1 since the horizontal segments are decreasing and non-overlapping. The middle and right plots violate Condition~1 since there are overlapping segments (circled).
}
\end{figure}

The process of 'scanning the scatter plot from top to bottom' can be formalized as follows (see Algorithm \ref{alg:condition1}).
Order the list of preferences 
lexicographically by two criteria: the primary criterion is descending frequency (as before), and the secondary criterion is ascending distance. 
So the preferences are partitioned to equivalence classes by their frequency: the classes are ordered by descending frequency, and within  each equivalence-class, the preferences are ordered by ascending distance from $\succ_0$. Preferences with both the same frequency and the same distance are ordered arbitrarily.
Under this ordering, the following lemma holds:
\begin{lemma}
\label{lem:cond1}
If Condition~1 is violated for \emph{any} pair of preference relations in $\pi$, then it is violated for an adjacent pair $\succ_{i},\succ_{i+1}$ for some $i$.
\end{lemma}
The lemma is easy to understand based on the graphical criterion outlined above. A formal proof is provided in Appendix \ref{sec:proof1}.

Lemma \ref{lem:cond1} implies that in order to ensure that Condition~1 is satisfied for all preferences in $\pi$, it is sufficient to scan the ordered list of preferences from $\succ_1$ to $\succ_{n'}$, and check if there is some $i$ such that $\freq{i} > \freq{i+1}$ and $\dist{i}\geq \dist{i+1}$.

If Condition~1 holds for all preferences in $\pi$, it remains to check that it holds for preferences not in $\pi$, i.e, preferences with zero frequency. Let 
$\widehat{d}$ be the largest distance from a preference in $\pi$ to $\succ_0$, i.e, $\widehat{d} = \dist{n'}$.
Condition~1  implies that, if $\mu_\pi(\succ) = 0$, then $d(\succ,\succ_0) >  \widehat{d}$. Therefore, we have to check that the distances between $\succ_{0}$ and preferences not in $\pi$ are all larger than $\widehat{d}$. Equivalently, we can ensure that all preferences with distance at most $\widehat{d}$ are in $\pi$. This can be checked by calculating the \emph{number} of possible preferences with distance at most $\widehat{d}$, and verifying that it is equal to the total number $n'$ of distinct preferences in $\pi$. Since this number involves all possible preferences, it does not depend on the candidate $\succ_0$. Therefore, we can calculate this number assuming w.l.o.g that $\succ_0$ is the preference defined by $1 \succ_0 2 \succ_0 \ldots \succ_0 K$. So we have to calculate the number of permutations on $K$ elements with at most $\widehat{d}$ \emph{inversions} (out-of-order elements). This can be written as:
\begin{align*}
\sum_{j=0}^{\widehat{d}} |T(K,j)|
\end{align*}

where $T(K,j)$ is the set of permutations on $K$ elements with \emph{exactly} $j$ inversions.%
\footnote{
$|T(K,j)|$ is also known as the \emph{Mahonian number}; see OEIS sequence A008302, https://oeis.org/A008302.
}

The number $|T(K,j)|$ can be calculated using the following recurrence relation%
\begin{itemize}
\item $\forall K: ~ |T(K,0)|=1$, since there is exactly one permutation with zero inversions --- the identity permutation.
\item $\forall j: ~ |T(0,j)|=0$, since there are no permutations with 0 elements.
\item $\forall K,j: ~ |T(K,j)| = \sum_{i=0}^{\min(K-1,j)} |T(K-1,j-i)|$:
for any permutation of $1,\ldots,K$ with $j$ inversions, let $i$ be the number of elements that come after element $K$ in that permutation. Since $K$ is larger than all other elements, there are exactly $i$ inversions involving $K$. Therefore, if we remove $K$, we get a permutation of $1,\ldots,K-1$ with exactly $j-i$ inversions. By summing the counts of these permutations for all possible values of $i$ (namely, $i\geq 0, i\leq j, i\leq K-1$) we get $|T(K,j)|$.
\end{itemize}

The algorithm for detecting Level-1 Consensus is summarized in Algorithm \ref{alg:level1consensus} in page \pageref{alg:level1consensus}.

We complete the proof of Theorem \ref{thm:level1consensus} by a run-time analysis.
\begin{lemma}
The run-time of Algorithm \ref{alg:level1consensus} is:
\begin{align*}
O(n'^2 K\sqrt{\log K} + n'^2 \log{n'})
\end{align*}
where $n'\leq \min(K!,n)$ is the number of distinct preferences in $\pi$.
\end{lemma}
\begin{proof}
We first analyze Algorithm \ref{alg:condition1}. It has to calculate the distance between $\succ_0$ and each of the other $n'-1$ preferences. Calculating the inversion distance between a given pair of preferences can be done by a recently-developed algorithm \cite{Chan2010Counting} with a runtime of $O(K\sqrt{\log K})$. We then have to  order the $n'$ distinct preferences and then scan them from top to bottom. Ordering $n'$ items can be done in time $O(n' \log {n'})$. The value of $n'$ is at most the maximum of $n$ (the number of voters) and $K!$ (the number of possible preferences). 
So the run-time of Algorithm \ref{alg:condition1} is 
\begin{align*}
O(n' K\sqrt{\log K} + n' \log{n'})
\end{align*}
As will be explained in the next section, the probability of having two preferences exhibiting exactly the same frequency is low, so in most cases we will have to 
apply Algorithm \ref{alg:condition1}
only once. However, in the improbable case in which there are many preferences with the same frequency, we would have to apply it at most $n'$ times. Therefore, the worst-case run-time of Algorithm
\ref{alg:level1consensus} is $n'$ times the run-time of Algorithm \ref{alg:condition1}.
\end{proof}


\section{Probability of Level-1 Consensus}
\label{sec:prob}
Equipped with a procedure for checking level-1 consensus, we set out to check how likely is this property in various settings. We conducted several simulation experiments.

In the first experiment we used the PrefLib database \citep{MaWa13a}, an online database of real-world preference-profiles collected from various sources. This database contains 315 full profiles, with different numbers of alternatives and voters; see Table \ref{tab:preflib} for statistics. For each profile, we used the algorithm described in the previous section to check whether there exists a 1-level consensus. The results were striking: none of the 315 profiles exhibited a level-1 consensus.

\begin{table}
\begin{tabular}{|c|c|c|c|c|}
\hline 
\textbf{Code} & \textbf{Description} & \textbf{\# profiles} & \textbf{\# alternatives} & \textbf{\# voters} \\ 
\hline 
\shortstack{ED-00004 1--100: \\Netflix Prize Data \cite{Bennett-Lanning-2007-Netflix}} &
\shortstack{Rankings of movies\\by consumers.} &
100 & 3 & 100--1000 \\ 
\hline 
\shortstack{ED-00004 101-200: \\Netflix Prize Data \cite{Bennett-Lanning-2007-Netflix}} &
\shortstack{Rankings of movies\\by consumers.} &
100 & 4 & 100--1000 \\ 
\hline 
\shortstack{ED-00006: \\Skate Data} &
\shortstack{Ranking of skaters\\by judges in competitions.} & 
20 & 10--25 & 8--10  \\ 
\hline 
\shortstack{ED-00009: \\AGH Course Selection} & 
\shortstack{Ranking of courses\\by university students.} & 
2 & 7--9 & $\approx150$ \\ 
\hline 
\shortstack{ED-00011: \\ Web Search} & 
\shortstack{Ranking of search-phrases\\by search-engines.} & 
3 & 100--250 & $5$  \\ 
\hline 
\shortstack{ED-00012: \\ T shirt} & 
\shortstack{Ranking of T-shirt designs\\by researchers.} & 
1 & 10 & $30$  \\ 
\hline 
\shortstack{ED-00014: \\ Sushi Data} & 
\shortstack{Ranking of sushi kinds\\by consumers.} & 
1 & 10 & $5000$  \\ 
\hline 
\shortstack{ED-00014: \\ Sushi Data} & 
\shortstack{Ranking of sushi kinds\\by consumers.} & 
1 & 10 & $5000$  \\ 
\hline 
\shortstack{ED-00015: \\ Clean Web Search} & 
\shortstack{Ranking of search-phrases\\by search-engines.} & 
79 & 10--250 & $4$  \\ 
\hline 
\shortstack{ED-00024: \\ Mechanical Turk Dots \cite{Mao2013Better}} & 
\shortstack{Ranking of dot-sets\\by Amazon-Turk workers.} & 
4 & 4 & $\approx 800$  \\ 
\hline 
\shortstack{ED-00025: \\ Mechanical Turk Puzzle  \cite{Mao2013Better}} & 
\shortstack{Ranking of puzzles\\by Amazon-Turk workers.} & 
4 & 4 & $\approx 800$  \\ 
\hline 
\shortstack{ED-00032: \\ Education Surveys} & 
\shortstack{Ranking of issues\\by informatics students.} & 
1 & 6 & $15$  \\ 
\hline 
\end{tabular} 
\caption{\label{tab:preflib} Summary of PrefLib \citep{MaWa13a} data-sets used in our experiments.}
\end{table}

In the second experiment we used preference-profiles that were generated according to \emph{Mallows' phi model} \citep{mallows1957}, which was claimed to favor level-1 consensus \citep{mahajne2015}. Mallows' model assumes that there is a ``correct'' preference $\succ_*$, and the actual preferences of the voters are noisy variants of it. The probability of a preference $\succ$ depends on its inversion distance from the correct preference: $d(\succ,\succ_*)$.  The strength of this dependence is determined by a parameter $\phi\in(0,1]$, where lower $\phi$ means higher dependence; when $\phi\to 0$ the preferences of all voters are identical and equal to $\succ_*$, while when $\phi=1$ the preference of each voter is selected uniformly at random from the $K!$ possible orderings on $K$ items. In general, the probability of each preference-relation $\succ$ is given by \citep{Tyler2014Mallows}:
\begin{align*}
\operatorname{Prob}\big[\succ~~ | ~~\phi,\succ_*\big] = {1\over Z}\cdot \phi^{d(\succ, \succ_*)}
\end{align*}
where $Z$ is a normalization factor.

We considered all 6 combinations of $K\in\{3,4,5\}$ alternatives and $n\in\{100,1000\}$ voters, where $\phi$ varied between 0 and 1. For each combination of $K,n,\phi$ we ran 1000 experiments and calculated (a) the fraction of profiles that exhibit level-1  consensus, (b) the fraction of profiles that are single-peaked,%
\footnote{
This calculation was done for the sake of comparison. It was implemented using Nicholas Mattei's PrefLib tools, which are freely available at GitHub: https://github.com/nmattei/PrefLib-Tools .
} 
and (c) the fraction of profiles that exhibit \emph{Flexible Consensus}, which will be presented in the next section. The results are shown in Figure \ref{fig:mallows}.
\begin{figure}
\hspace*{-2cm}
\includegraphics[width=1.2\textwidth]{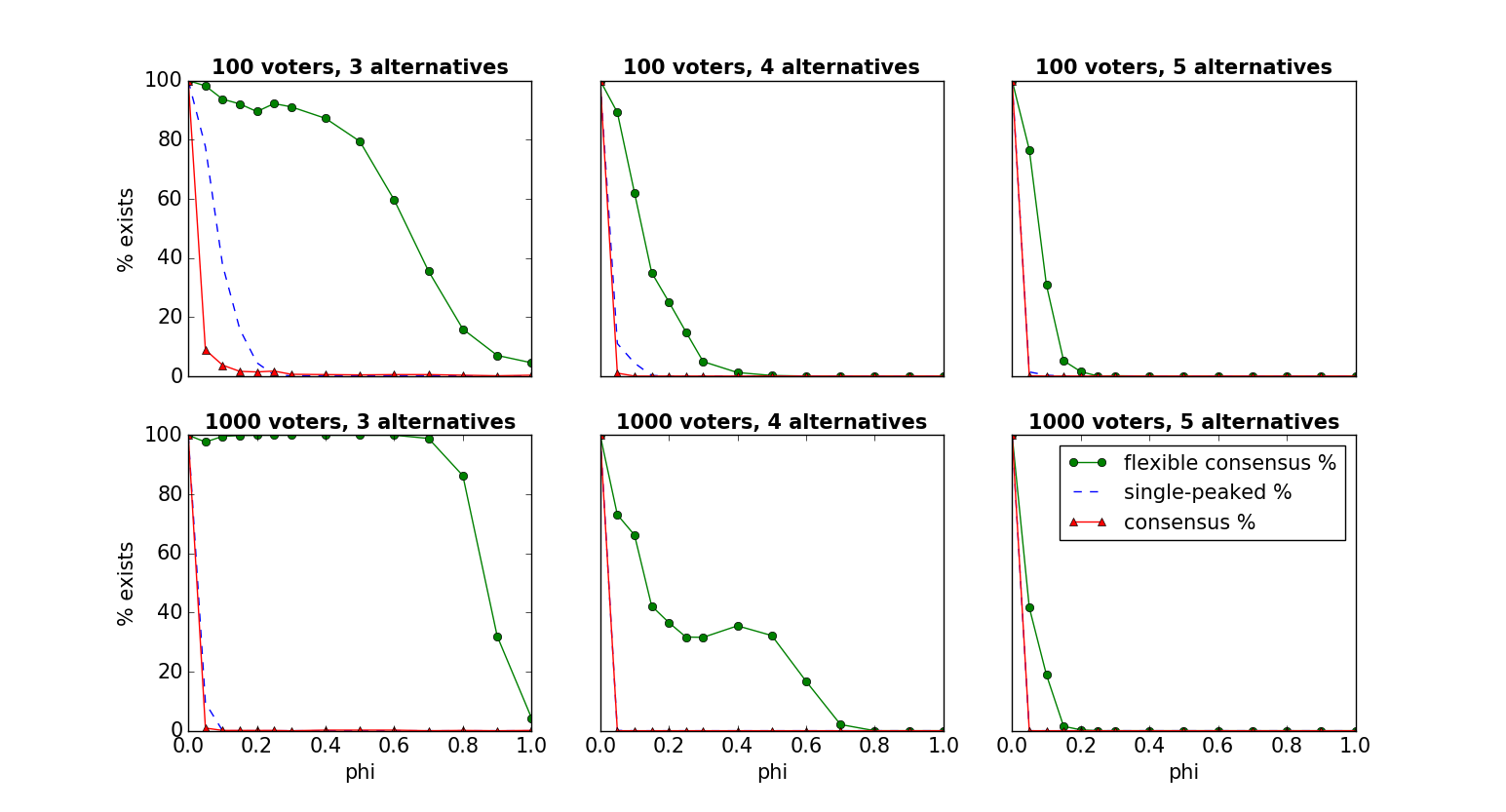}
\caption{
\label{fig:mallows}
Percentage of profiles, from a set of profiles selected at random according to Mallows' phi model, which exhibit level-1-consensus, single-peakedness, or Flexible Consensus (the latter is defined in Section \ref{sec:flexible}).
}
\end{figure}
As can be expected, consensus always exists when $\phi=0$, since in this case there is a deterministic consensus on the ``true'' preference. Additionally, when there are 100 voters and 3 alternatives and $\phi$ is small, a small positive percentage of the profiles exhibit a level-1 consensus (top left plot). In all other cases, the percentage of level-1 consensus profiles drops to 0 when $\phi\geq 0.05$.

Why is level-1 consensus so rare? 

Intuitively, the reason is that it requires groups of preferences to have exactly the same frequency in the population. Condition~1 implies that if $\dist{i}=\dist{j}$ then $\freq{i}=\freq{j}$. For every $K\geq 3$ and for every candidate $\succ_0$, there exist at least two preferences with the same distance from the candidate, $\dist{i}=\dist{j}$. Hence, a necessary condition for level-1 consensus is that there exist at least two preferences with exactly the same frequency. 
As the number of voters goes to $\infty$, 
the probability that any specific preference-relation appears with any specific frequency goes to 0. Therefore, the probability that two preference-relations have the same frequency goes to 0 too.

To formalize this intuition, we present an asymptotic upper bound on the probability of level-1 consensus for the case $\phi=1$. 
This is the case of impartial culture, in which all $K!$ preferences are equally probable. 

We select a preference-profile $\pi$ according to the following random process, parametrized by $K$ (the number of alternatives) and an integer-valued parameter $m$.
\begin{itemize}
\item Let $(\succ_i)_{i=1}^{K!}$ be an enumeration of the preferences in $\prefs$. 
\item 
For each $i$, draw a number $B_i$ according to a binomial distribution with $m$ trials and success-probability $1/K!$.
\item 
Return a profile in which, for every  $i$, there are $B_i$ voters whose preference is $\succ_i$.
\end{itemize}
All the $B_i$'s are i.i.d. random variables with mean value $\mu := {m\over K!}$. 
The total number of preferences in the profile $\pi$ is $n = \sum_{i=1}^{K!}B_i$; this is also a random variable, and its mean value is $E[n] = \sum_{i=1}^{K!}{m\over K!} = m$.
The process is symmetric with respect to the preferences in $\prefs$, so all preferences are equally likely, in accordance with the impartial-culture assumption.

Define $\pcon(m,K)$ as the probability that the above random process yields a profile that exhibits Level-1 Consensus. The rest of the section is devoted to proving the following proposition:
\begin{proposition}
\label{prop:pconsensus}
For every $K$:
\begin{align*}
\lim_{m\to \infty} \pcon(m,K) = 0
\end{align*}
\end{proposition}
\begin{proof}
We first show an upper bound on the probability for level-1 consensus around a fixed preference.

Let $\succ_0$ be a fixed preference. 
Denote by $\pconfixed(m,K)$ the probability that a profile $\pi$, selected according to the random process above with parameters $m,K$, exhibits Level-1 Consensus around $\succ_0$.
Since all preferences are equally probable, $\pconfixed(m,K)$ is the same for all $\succ_0$. We now present an approximate upper bound on $\pconfixed(m,K)$.

Recall that $T(K,d)$ is the set of distinct preferences whose inversion-distance from $\succ_0$ is exactly $d$,
where $d$ can vary between 0 and ${K\choose 2}$ (the number of pairs of $K$ alternatives). 
The conditions for Level-1 Consensus imply that all preferences in $T(K,d)$ must have the same frequency in $\pi$, i.e:
\begin{align}
\label{eq:pconfixed-1}
\text{For all~} i,j
\text{~~such that~~}  \succ_i \in T(K,d)
\text{~~and~~} \succ_j \in T(K,d):
~~ 
B_i=B_j
\end{align}
Let $\peq(m,p,t)$ be the probability that $t$~ i.i.d. random variables distributed like $Binom[m,p]$ are equal. Then, (\ref{eq:pconfixed-1}) implies that:
\begin{align}
\label{eq:pconfixed-2}
\pconfixed(m,K)
\leq
\prod_{d=0}^{K\choose 2} \peq\left(m,{1\over K!},|T(K,d)|\right)
\end{align}
In Appendix \ref{sec:binomial}, we show that $\peq(m,p,t)$ can be approximated as:
\begin{align*}
\peq(m,p,t)
&\approx
{1
\over 
(\sqrt{2\pi p q m})^{t-1}
}
\end{align*}
where $q=1-p$ and the symbol $\approx$ means that the ratio of the expressions in its two sides goes to $1$ as $m\to\infty$. Therefore, (\ref{eq:pconfixed-2}) can be approximated as:
\begin{align}
\pconfixed(m,K)
\leq&\approx 
{1 / 
\sqrt{
\prod_{d=0}^{K\choose 2}
\left(2\pi {1\over K!} m\right)^{[|T(k,d)|-1]} 
}
}
\notag
\\
\label{eq:pconfixed-3}
&=
{1 / 
\sqrt{
\left({2\pi m\over K!} \right)^{\sum_{d=0}^{K\choose 2} [|T(k,d)|-1] } 
}
}
\end{align}
The sum in the exponent can be simplified as follows:
\begin{align*}
\sum_{d=0}^{K\choose 2} [|T(k,d)|-1]
=
\left[\sum_{d=0}^{K\choose 2} |T(k,d)|\right] - {K\choose 2} - 1
\end{align*}
The sum of $|T(K,d)|$ over all possible values of $d$ (i.e, $d$ between 0 and ${K\choose 2}$) equals the total number of different preferences over $K$ items, which is $K!$:
\begin{align}
\sum_{d=0}^{K\choose 2} |T(k,d)| = K!
\end{align}
Substituting in (\ref{eq:pconfixed-3}) gives:
\begin{align*}
\pconfixed(m,K)
\leq&\approx 
{1 / 
\sqrt{
\left(2\pi m\over K!\right)^{K!-{K\choose 2}-1} 
}
}
\end{align*}

Recall that $\pconfixed(m,K)$ is the probability of level-1 consensus around a \emph{fixed} preference. The probability of level-1 consensus around \emph{any} preference is, by the union bound, at most this probability times the number of possible preferences, i.e, 
\begin{align*}
\pcon(m,k)
&\leq
K! \cdot \pconfixed(m,K) 
\\
&\leq 
\frac{K!}{
\sqrt{
\left(2\pi m\over K!\right)
^{K!-{K\choose 2}-1}}}
\end{align*}
so for every fixed $K\geq 3$, $\lim_{m\to\infty}\pcon(m,k) = 0$.
\end{proof}
As an illustration, with $K=3$ alternatives we get an upper bound of $6/\sqrt{(2\pi m/6)^2}\in O(1/m)$; with $K=4$ alternatives the upper bound is $24/\sqrt{(2\pi m/24)^{17}}\in O(1/m^{8.5})$. The rate of convergence to zero is faster when $K$ is larger.

\section{Flexible Consensus}
\label{sec:flexible}
Motivated by the low probability of a level-1 consensus, we suggest below a weakened variant of this property termed \emph{Flexible Consensus}.
It is equivalent to a weakened version of Condition~1, without Condition 2.

\begin{definition}
\label{def:flexible-consensus}
Let $\succ_0 \in \prefs$.
A profile $\pi \in \prefs^n$ exhibits \emph{Flexible Consensus around $\succ_0$} if the following condition holds:
\begin{quote}
(Flexible Condition~1)

For all $\succ',\succ$,~~~if~~$\mu_\pi(\succ') >  \mu_\pi(\succ)$,~~~
then~~$d(\succ', \succ_0)  \leq  d(\succ, \succ_0)$.
\end{quote}
\end{definition}

The only difference between Condition 1 and Flexible Condition 1 is that $d(\succ', \succ_0) \ < \ d(\succ, \succ_0)$ is replaced by $d(\succ', \succ_0) \ \leq \ d(\succ, \succ_0)$. 
It will be shown below that this apparently minor change 
significantly increases the probability that the condition is satisfied, while keeping the desirable stability properties of the original condition.
Moreover, these stability properties hold even without Condition~2.\footnote{
Our proofs below closely follow the proofs of \cite{mahajne2015}. Their proofs are stated for level-$r$ consensus for general $r$, and indeed Flexible Condition~1 can also be adapted to general $r$, but for the sake of simplicity we prefer to focus on the case $r=1$.
}

\subsection{First stability property: Existence of a Condorcet winner}
\begin{definition}
Given a profile $\pi$, the \emph{majority relation} $M_\pi$ is defined as follows: $a M_\pi b$ iff, in a vote between alternatives $a$ and $b$, $a$ beats $b$ by a weak majority. I.e, the number of preferences in $\pi$ by which $a\succ b$ is at least as large as the number of preferences in $\pi$ by which $b\succ a$.
\end{definition}

\begin{lemma}
\label{lem:cond1'}
Let $\pi\in\prefs^n$ be a profile of $n$ voters and $\succ_0\in \prefs$ a preference-relation such that Flexible Condition~1 is satisfied. 
~~~
Then, for any two alternatives $a,b$:

(1) If $a\succ_0 b$ ($a$ is preferred to $b$ by $\succ_0$)
then $a M_\pi b$.

(2) If $n$ is odd then the opposite is also true: if $a M_\pi b$ then $a\succ_0 b$.

(3) If $n$ is even then the opposite is ``almost'' true: if $a M_\pi b$ then either $a\succ_0 b$ or $a\succ_1 b$, where $\succ_1$ is another preference around which there is Flexible Consensus.\footnote{We are grateful to an anonymous referee for suggesting this additional condition.}
\end{lemma}
\begin{proof}
(1) Suppose that $a\succ_0 b$.
Partition $\prefs$, the set of $K!$ possible preferences, to two subsets: 
\begin{itemize}
\item The subset $C(a>b)$ containing the $K!/2$ preferences for which $a\succ  b$;
\item The subset $C(b>a)$ containing the $K!/2$ preferences for which $b\succ  a$.
\end{itemize}
Let $\switch: C(b>a) \to C(a>b)$ be the bijection that takes a preference in $C(b>a)$ and switches $a$ with $b$ in the ranking. Since $a\succ_0 b$, 
this switch brings every preference in $C(b>a)$ at least one step closer to $\succ_0$. I.e, 
for every preference $\succ\in C(b>a)$ it holds that $d(\switch(\succ),\succ_0) < d(\succ,\succ_0)$ (see proof in Appendix \ref{sec:switch}).

By Flexible Condition~1, this implies that $\mu_\pi(\switch(\succ)) \geq \mu_\pi(\succ)$. So for every preference $\succ$ by which $b$ is preferred to $a$ corresponds a unique preference $\switch(\succ)$ by which $a$ is preferred to $b$, which is at least as frequent. Therefore, $a$ beats $b$ by a weak majority: $a M_\pi b$.

(2) When $n$ is odd and  $a M_\pi b$, the majority must be strict, so it is not true that $b M_\pi a$. Hence, by (1), it is not true that $b \succ_0 a$. 
By our assumption, $\succ_0$ is a strict total order. Therefore, $a\succ_0 b$.

(3) When $n$ is even and $a M_\pi b$, 
there are two cases: if $a\succ_0 b$ then we are done.
If $b\succ_0 a$, then 
$a \succ_1 b$, where $\succ_1 = \switch(\succ_0)$.
It remains to prove that there is Flexible Consensus around $\switch(\succ_0)$.

Step I. Since $b\succ_0 a$, the 
argument in (1) shows that, 
for every preference $\succ\in C(a>b)$, 
the frequency 
$\mu_\pi(\switch(\succ)) \geq \mu_\pi(\succ)$.
But since $a M_\pi b$, all these inequalities must be equalities, i.e, 
for every preference $\succ\in C(a>b)$, 
we must have
$\mu_\pi(\switch(\succ)) = \mu_\pi(\succ)$.

Step II. 
For every two preferences $\succ,\succ'$:
$d(\succ,\succ') = d(\switch(\succ),\switch(\succ'))$.
To see this, consider the pairs of alternatives that are inversed between $\succ,\succ'$. If, in each such pair, we replace $a$ by $b$ and $b$ by $a$, then we get exactly the pairs of alternatives that are inversed between 
$\switch(\succ)$ and $\switch(\succ')$.

Step III. Let $\succ',\succ$ be two preferences for which $d(\succ',\switch(\succ_0)) < d(\succ,\switch(\succ_0))$.
Then, by Step II, 
$d(\switch(\succ'),\succ_0) < d(\switch(\succ),\succ_0)$.
Since there is Flexible Consensus around $\succ_0$, this implies:
$\mu_\pi(\switch(\succ')) \geq  \mu_\pi(\switch(\succ))$.
By Step I, this implies:
$\mu_\pi(\succ') \geq  \mu_\pi(\succ)$.
Hence, there is Flexible Consensus around $\switch(\succ_0)$. 

\end{proof}

\begin{definition}
For every preference $\succ_0\in \prefs$, 
the alternative ranked first according to $\succ_0$ is denoted by $Best(\succ_0)$ 
\end{definition}

\begin{definition}
Given a profile $\pi\in \prefs^n$,
a \emph{weak Condorcet winner} of $\pi$ is an alternative $a$ that beats all other alternatives by a weak majority, i.e, for any other alternative $b$, $a M_\pi b$.
\end{definition}

\begin{theorem}
\label{thm:condorcet}
Let $\pi\in\prefs^n$ be a profile and $\succ_0\in \prefs$ a preference-relation around which there is Flexible Consensus. Then $Best(\succ_0)$ is a weak Condorcet winner of $\pi$.

Moreover, if $n$ is odd then $\succ_0$ coincides with the majority relation $M_\pi$, and $\succ_0$ is the unique preference in $\prefs$ for which Flexible Condition~1 is satisfied.
\end{theorem}
\begin{proof}
Let $a := Best(\succ_0)$. So for every $b\neq a$, $a\succ_0 b$. By Lemma \ref{lem:cond1'}, this implies that $a M_\pi b$. Hence, $a$ is a weak Condorcet winner of $\pi$.

When $n$ is odd, Lemma \ref{lem:cond1'} implies that $a\succ_0 b$ iff $a M_\pi b$, so $\succ_0$ coincides with the ordering induced by $M_\pi$. This is true for any preference in $\prefs$ for which Flexible Condition~1 holds, so any such preference coincides with $\succ_0$.
\end{proof}

\subsection{Second stability property: agreement of scoring rules}

A \emph{scoring rule} is a rule for selecting an alternative based on a profile. 
\begin{definition}
A \emph{scoring rule} is a rule characterized by a vector $S$ of length $K$, $S_1\geq \cdots \geq S_K$. Given a profile $\pi$, for each preference $\succ\in \pi$, the rule assigns score $S_1$ to the alternative ranked first by $\succ$, $S_2$ to the alternative ranked second by $\succ$, and so on. The rule then sums the scores assigned to each alternative by all preferences in $\pi$, and selects the alternative/s that received the highest total score.
\end{definition}
In general, every scoring-rule $S$ might select a different alternative. But below we show that, if a profile exhibits Flexible Consensus, then there is an alternative which is selected all scoring rules.

\begin{lemma}
\label{lem:cond1'score}
Let $\pi\in\prefs^n$ be a profile and $\succ_0\in \prefs$ a preference-relation such that Flexible Condition~1 is satisfied. 
~~~
Then, for any two alternatives $a,b$ and any scoring-rule $S$, if $a\succ_0 b$ then the score of $a$ is at least as large as the score of $b$.
\end{lemma}
\begin{proof}
Similarly to the proof of Lemma \ref{lem:cond1'}, we partition $\prefs$ into two subsets, $C(a>b)$ and $C(b>a)$, and define the bijection $\switch$ between them.

For every scoring-rule $S$ and preference $\succ\in\prefs$, define $\Delta_{S,a,b}(\succ)$ as the difference between the score of $a$ in $\succ$ and the score of $b$ in $\succ$. By definition of a scoring rule:
\begin{itemize}
\item For every preference $\succ\in C(a>b)$, $\Delta_{S,a,b}(\succ)$ is weakly-positive.
\item For every preference $\succ\in C(b>a)$, $\Delta_{S,a,b}(\succ)$ is weakly-negative.
\item For every preference $\succ\in \prefs$, $\Delta_{S,a,b}(\succ) = - \Delta_{S,a,b}(\switch(\succ))$.
\end{itemize}
Given the scoring rule $S$ and the profile $\pi$, define $\Delta_{S,a,b}(\pi)$ as the difference between the total score of $a$ in $\pi$ and the total score of $b$ in $\pi$. Then, by definition:
\begin{align*}
\Delta_{S,a,b}(\pi) 
&= 
\sum_{\succ\in\prefs}
\mu_\pi(\succ)\cdot \Delta_{S,a,b}(\succ)
\\
&=
\sum_{\succ\in C(a>b)}
\mu_\pi(\succ)\cdot \Delta_{S,a,b}(\succ)
+
\sum_{\succ\in C(b>a)}
\mu_\pi(\succ)\cdot \Delta_{S,a,b}(\succ)
\\
&=
\sum_{\succ\in C(a>b)}
\left[
\mu_\pi(\succ)\cdot \Delta_{S,a,b}(\succ)
+
\mu_\pi(\switch(\succ))\cdot \Delta_{S,a,b}(\switch(\succ))
\right]
\\
&=
\sum_{\succ\in C(a>b)}
\left[
\mu_\pi(\succ)\cdot \Delta_{S,a,b}(\succ) - 
\mu_\pi(\switch(\succ))\cdot \Delta_{S,a,b}(\succ)
\right]
&&
\\
&=
\sum_{\succ\in C(a>b)}
\Delta_{S,a,b}(\succ) \cdot \left[
\mu_\pi(\succ) - 
\mu_\pi(\switch(\succ))
\right]
\end{align*}
Since $a\succ_0 b$, for every preference $\succ\in C(a>b)$, the lemma in Appendix \ref{sec:switch} implies that $d(\switch(\succ),\succ_0) > d(\succ,\succ_0)$. Hence, by Flexible Condition~1, $\mu_\pi(\succ) \geq 
\mu_\pi(\switch(\succ)$. Hence, all terms in the last sum are weakly-positive. Hence, $\Delta_S(\pi)\geq 0$ and the lemma is proved.
\end{proof}
\begin{theorem}
Let $\pi\in\prefs^n$ be a profile and $\succ_0\in \prefs$ a preference-relation such that Flexible Condition~1 is satisfied. Then 
the score assigned to $Best(\succ_0)$ by every scoring rule is at least as high as the score assigned to any other alternative by the same rule.
\end{theorem}
\begin{proof}
Follows immediately from Lemma \ref{lem:cond1'score}.
\end{proof}

\section{Algorithm for Detecting Flexible Consensus}

\begin{figure}
\begin{algorithm}[H]
\begin{flushleft}
INPUT: 
\begin{itemize}
\item A set of $K$ alternatives, $A = \{a_1,\ldots,a_K\}$.
\item A profile $\pi$ containing $n$ preference-relations on $A$ (possibly with duplicates).
\end{itemize}

OUTPUT:
\begin{itemize}
\item If $\pi$ exhibits Flexible Consensus around some preference $\succ_0$, return  $\succ_0$. 
\item Otherwise, return ``no consensus''.
\end{itemize}

ALGORITHM:
\begin{enumerate}
\item Calculate the frequency $\mu_\pi(\succ)$ of each preference $\succ\in\pi$.
\item Define $n'=$ number of distinct preferences in $\pi$.
\item Order the preferences by descending frequency:
$\mu_\pi(\succ_1) ~\geq~ \mu_\pi(\succ_2) ~\geq~ \cdots~\geq~ \mu_\pi(\succ_{n'})$.
\item Set $M := \freq{1} $ to be the maximum frequency of a preference relation.
\item For $j = 1, 2, \ldots n'$ \  while $\mu_\pi(\succ_j) = M$: 
Check if Flexible Condition~1 is satisfied for $\pi$ and  $\succ_j$ (see Algorithm \ref{alg:flexible-condition1}). If it is, then return $\succ_j$.
\item Return ``no consensus; Flexible Condition~1 is violated for all candidates''.
\end{enumerate}
\end{flushleft}
\caption{
\label{alg:flexible-level1consensus}
Detection of Flexible Consensus 
}
\end{algorithm}
\vskip -8mm
\end{figure}

\label{sec:algorithm-1-flexible}
Checking the existence of Flexible Consensus is very similar to checking Level-1 Consensus. The check is presented in Algorithm \ref{alg:flexible-level1consensus}. 
It is very similar to Algorithm \ref{alg:level1consensus}; the only differences are that we do not need to check Condition 2 (since it does not exist in Flexible-Level-1-Consensus), and instead of checking Condition 1 for each candidate, we check Flexible-Condition-1.

The procedure for checking Flexible Condition~1
is presented as Algorithm \ref{alg:flexible-condition1}.
It is very similar to the one for checking Condition~1 in Algorithm \ref{alg:condition1}. 
There are two differences: the inequality that leads to the failure of the procedure is $\dist{i} > \dist{i+1}$ (instead of $\dist{i} \geq \dist{i+1}$), 
and in the last step we have to check that 
there is no preference outside of $\pi$
whose distance to $\succ_0$ is less than  $\widehat{d}$ (instead of less-than-or-equal-to $\widehat{d}$). 
Hence, by following the same proof as in Section \ref{sec:algorithm-1}, it is easy to prove:
\begin{theorem}
\label{thm:flexible}
Algorithm \ref{alg:flexible-level1consensus} checks whether a profile $\pi$ exhibits Flexible Consensus in time:
\begin{align*}
O(n'^2 K\sqrt{\log K} + n'^2 \log{n'})
\end{align*}
where $n'=\min(K!,n)$.
\end{theorem}

\begin{figure}
\begin{algorithm}[H]
\begin{flushleft}
INPUT: 
\begin{itemize}
\item A profile $\pi$ containing $n$ preference-relations on $A$ (possibly with duplicates).
\item A preference $\succ_0$.
\end{itemize}

OUTPUT: ``True'' if Flexible Condition~1 is satisfied for $\pi$ and $\succ_0$. ``False'' otherwise.

ALGORITHM:

\begin{enumerate}
\item Order preferences with the same frequency by ascending inversion-distance from $\succ_0$.
\item For $i = 1, 2, \ldots, n'-1$:
\\
If $\freq{i} ~ > ~ \freq{i+1}$  and  $\dist{i} > \dist{i+1}$,
 return False.
\item Set $\widehat{d}:=\dist{n'}$. Define $n^*$ as the number of profiles in $\pi$ whose distance to $\succ_0$ is at most $\widehat{d}-1$.
\hskip 1cm
If $\sum_{j=0}^{\widehat{d}-1} T(K,j) = n^*$, return True.
\hskip 1cm
Else, return False.
\end{enumerate}
\end{flushleft}
\caption{
\label{alg:flexible-condition1}
Check if Flexible Condition~1 is satisfied
}
\end{algorithm}
\end{figure}

\section{Probability of Flexible Consensus}
We applied Algorithm \ref{alg:flexible-level1consensus} to the same experimental settings described in Section \ref{sec:prob} and estimated the probability of having Flexible Consensus. Out of the 315 PrefLib profiles, \textbf{39} exhibit Flexible Consensus. All 39 profiles are from the dataset labeled ``ED-00004 1--100'', where all profiles have 3 alternatives. This means that 39\% of all these profiles with 3 alternatives exhibited the Flexible Consensus (in contrast to \textbf{0} which exhibited level-1 consensus). 

The results of the experiments on random profiles are shown in Figure \ref{fig:mallows} in page \pageref{fig:mallows}; it is evident that in all settings, including the most challenging setting of impartial culture ($\phi=1$), Flexible Consensus is substantially more probable than level-1 consensus.

We complement the empirical findings with a theoretical analysis of the asymptotic probability of Flexible Consensus under the impartial culture assumption. 
We consider a profile generated by the random process 
described before Proposition \ref{prop:pconsensus}
in page \pageref{prop:pconsensus}. 
We denote by 
$\pfcon(m,K)$ 
the probability that such a random profile  exhibits Flexible Consensus. 
In stark contrast to Proposition \ref{prop:pconsensus}, we prove that this probability is always larger than a positive constant, even when the number of voters goes to infinity:
\begin{proposition}
\label{prop:pflexibleconsensus}
For every $K$ there exists a constant $C_K>0$ such that:
\begin{align*}
\forall m: \pfcon(m,K) \geq C_K
\end{align*}
\end{proposition}
\begin{proof}
As explained in the proof of Proposition \ref{prop:pconsensus}, the probability that two or more preferences have exactly the same frequency goes to 0 when $m\to\infty$, so for simplicity we neglect this possibility and assume that each preference relation has a different frequency.   
Note that this assumption can only decrease the probability of Flexible Consensus, since it implies that there is 
a unique preference with maximum frequency,
so there is only one candidate that can possibly satisfy Flexible Condition 1.
We denote this candidate by $\succ_0$.
Below we calculate the probability that Flexible Condition~1 holds for this preference. 

For every $i \geq 1$, define:
\begin{align*}
F_i ~ := ~\{~\mu_\pi(\succ)~ |~ d(\succ,\succ_0) =  i~\}
\end{align*}
so $F_i$ contains the frequencies of all preferences whose distance from $\succ_0$ is exactly $i$. 
Note that $F_i$ is non-empty only when $i\leq {K\choose 2}$, since ${K\choose 2}$ is the maximum possible inversion-distance between two preferences on $K$ alternatives.

Flexible Condition~1 is equivalent to the requirement that each member of the set $F_i$ is larger than each member of the set $F_j$, for every $i<j$.
\footnote{
In fact, it is equivalent to 
the requirement that each member of $F_i$ 
is larger than \emph{or equal to} each member of $F_j$, 
but as explained above, we neglect the possibility that two preferences have equal frequency.
} 
However, it does not impose any restrictions on the frequencies \emph{within} $F_i$.

Let $F := \cup_{i} F_i = $ the set of frequencies of the $K!-1$ preferences different than $\succ_0$. The total number of different orders on $F$ is $|F|!$. The total number of orders that satisfy Flexible Condition~1 is $|F_1|!\cdot|F_2|!\cdots|F_{K\choose 2}|!$. Since all preferences are equally likely, all $|F|!$ orders are equally likely. Therefore, the probability that the order of frequencies satisfies Flexible Condition~1 is at least:
\begin{align}
\label{eq:cond1'}
\frac{
|F_1|! |F_2|! \cdots|F_{K\choose 2}|!
}
{
|F|!
}
\end{align}
which is a positive constant that does not approach 0 even when $m\to\infty$.
\end{proof}

As an illustration, we calculate the lower bound for $K=3$ alternatives. In this case we have $|F_1|=2$ and $|F_2|=2$ and $|F_3|=1$ and $|F|=2+2+1=5$. Therefore, the probability that Flexible Condition~1 is satisfied in a random impartial-culture profile is at least $2!\cdot 2! \cdot1! / 5! = 1/30\approx 0.033$. 
In our experiments with $\phi=1$, the fraction of profiles with Flexible Consensus (in 1000 experiments) was $0.045$ for 100 voters and $0.043$ for 1000 voters. This is slightly higher than the lower bound of $0.033$, which can be explained by the fact that, when $n$ is finite, there is a positive probability that two preferences have the same frequency.

When $K>3$, the asymptotic probability of Flexible Condition~1 in an impartial culture remains positive, though much lower. For example, for $K=4$ the lower bound is approximately $10^{-12}$. 
This is consistent with the fact that we found no profiles that exhibit Flexible Consensus in our experiments with $\phi=1$ and $K\geq 4$.

\section{Conclusion}
We presented a practical procedure for checking whether a preference-profile exhibits a level-1 consensus. 
We realized that this property is highly improbable, and found a weaker property, Flexible Consensus, that preserves the desirable stability properties of the social choice. Furthermore, Flexible Consensus is provably more likely than level-1 consensus. 
This was demonstrated both theoretically (for the impartial culture setting) and empirically and over a database of real-world profiles.

Our experiments can be reproduced by re-running the code, which is freely available through the following GitHub fork:
https://github.com/erelsgl/PrefLib-Tools .

\begin{acknowledgments}
We are grateful to the anonymous referees for insightful comments that led to this much improved version of the manuscript, and to Marc Timme and Yuval Benjamini for helpful and valuable discussions.

We are thankful to 
several users of the math.stackexchange.com website, in particular 
Hernan J. Gonzalez, Vineel Kumar Reddy Kovvuri, spaceisdarkgreen and zhw.,  for their kind help in coping with mathematical issues.

M.N. is grateful to the Azrieli Foundation for the award of an Azrieli Fellowship. 
E.S. is grateful to the Israeli Science Foundation for the ISF grant 1083/13.
\end{acknowledgments} 

\appendix
\section*{APPENDICES}
\section{To verify Condition 1, it is sufficient to check adjacent preferences}
\label{sec:proof1}
This section provides a formal proof to the following lemma used in subsection \ref{sec:algorithm-1}.

\noindent
\textbf{Lemma \ref{lem:cond1}.}
Suppose that the preference relations in a profile $\pi$ are ordered by two criteria: first by frequency $\freq{i}$, then by distance $\dist{i}$, where $\succ_0$ is a fixed preference. In this ordering, if Condition~1 is violated for \emph{any} pair of preference relations in $\pi$, then it is violated for an adjacent pair $\succ_{i},\succ_{i+1}$ for some $i$.
\begin{proof}
	Suppose that there exist indices $i<j$ such that Condition~1 is violated for the pair $\succ_i,\succ_j$, i.e,  $\freq{i}>\freq{j}$, yet $\dist{i}\geq \dist{j}$. We now prove the lemma by induction on the difference of ​indices, $j-i$.
	
	\textbf{Base}: If $j-i=1$, then  $\succ_i$ and $\succ_j$ are already adjacent, so we are done.
	
	\textbf{Step}: Suppose that $j-i>1$. We prove that there exists a pair with a smaller difference that violates Condition~1. We consider several cases.
	
	\emph{Case \#1}: $i$ is not the largest index in its equivalence class. i.e, there exists $i'> i$ such that $\freq{i'}=\freq{i}$. Then, by the secondary ordering criterion, $\dist{i'}\geq \dist{i}$,  Condition~1 is violated for the pair $\succ_{i'}$ and $\succ_{j}$.
	
	\emph{Case \#2}: $j$ is not the smallest index in its equivalence class. i.e, there exists $j'<j$ such that $\freq{j'}=\freq{j}$. Then, by the secondary ordering criterion, $\dist{j}\geq \dist{j'}$,  Condition~1 is violated for the pair $\succ_{i}$ and $\succ_{j'}$.
	
	Otherwise, $i$ is the largest index in its equivalence class, $j$ is the smallest index in its equivalence class, but still $i+1<j$. This means that the equivalence classes of $i$ and $j$ are not adjacent, i.e, $\freq{i}>\freq{i+1}>\freq{j}$. Now there are two remaining cases:
	
	\emph{Case \#3}: $\dist{i}\geq \dist{i+1}$, in which case Condition~1 is violated for the adjacent pair $\succ_{i}$ and $\succ_{i+1}$ and we are done.
	
	\emph{Case \#4}: $\dist{i+1} > \dist{i}$. This implies $\dist{i+1} > \dist{j}$, so Condition~1 is violated for the pair $\succ_{i+1}$ and $\succ_{j}$ and we are done.
\end{proof}

\section{The effect of a switch on the inversion-distance}
\label{sec:switch}
This section provides a formal proof to an intuitive claim made within the proof of Lemma \ref{lem:cond1'}.
\hskip 1cm
Let $a,b$ be two fixed alternatives. Let $C(a>b)$ be the set of preferences by which $a\succ b$ and $C(b>a)$ the set of preferences by which $b\succ a$. Let $\switch:C(b>a)\to C(a>b)$ be a function that takes a preference-relation and creates a new preference-relation by switching the position of $a$ and $b$ in the ranking. 

\begin{lemma}
If $a \succ_0 b$, 
then for every preference $\succ_1\in C(b>a)$:
\begin{align*}
d(\succ_1,\succ_0)
~ > ~ 
d(\switch(\succ_1 ),\succ_0)
\end{align*}
\end{lemma}
\begin{proof}
For every preference $\succ$, define $D(\succ,\succ_0)$ as the set of pairs-of-alternatives $\{i,j\}$ that are ranked differently in $\succ$ and in $\succ_0$. By definition, the inversion distance is the cardinality of this set, $
d(\succ,\succ_0) = |D(\succ,\succ_0)|$.
Therefore, it is sufficient to show that there are more pairs in $D(\succ_1,\succ_0)$ than in $D(\switch(\succ_1),\succ_0)$. To show this, we consider all possible pairs-of-alternatives; for each pair, we calculate its contribution to the difference in cardinalities $|D(\succ_1,\succ_0)|-|D(\switch(\succ_1),\succ_0)|$, and show that the net difference is positive.
\begin{itemize}
\item The pair $\{a,b\}$ is in $D(\succ_1,\succ_0)$ but not in $D(\switch(\succ_1),\succ_0)$, so this pair contributes $+1$ to the difference.
\item Any pair that contains neither $a$ nor $b$ is not affected by the switch. I.e, each pair $\{c,e\}$ where $c,e\neq a,b$ is in $D(\succ_1,\succ_0)$ if-and-only-if it is in $D(\switch(\succ_1),\succ_0)$, so it contributes $0$ to the difference.
\item 
Let $c$ be an alternative that is ranked by $\succ_1$ above $b$ or below $a$, i.e, either
$c\succ_1 b\succ_1 a$ or $b\succ_1 a\succ_1 c$. 
Then, the order between $c$ to $a$ or $b$ in $\succ_1$ is not affected by the switch, so $\{c,a\}$ 
is in 
$D(\succ_1,\succ_0)$ if-and-only-if it is in $D(\switch(\succ_1),\succ_0)$, so it contributes $0$ to the difference. The same is true for $\{c,b\}$.
\item 
Finally, let $c$ be an alternative that is ranked by $\succ_1$ between $a$ and $b$, i.e, 
$b\succ_1 c\succ_1  a$. 
Then, the switch $\switch$ changes the direction of both the pair $\{c,a\}$ and the pair $\{c,b\}$:
\begin{itemize}
\item If $c \succ_0 a \succ_0 b$, then the pair $\{c,a\}$ is in $D(\succ_0,\switch(\succ_1))$ but not in $D(\succ_0,\succ_1)$, and the pair $\{c,b\}$ is in $D(\succ_0,\succ_1)$ but not in $D(\succ_0,\switch(\succ_1))$, so these pairs contribute $+1-1=0$ to the difference.
\item If $a \succ_0 c \succ_0 b$, then both the pair $\{c,a\}$ and the pair $\{c,b\}$ are in $D(\succ_0,\succ_1)$ but not in $D(\succ_0,\switch(\succ_1))$, so these pairs contribute $+1+1=2$ to the difference.
\item If $a \succ_0 b \succ_0 c$, then the pair $\{c,a\}$ is in $D(\succ_0,\succ_1)$ but not in $D(\succ_0,\switch(\succ_1))$, and the pair $\{c,b\}$ is in $D(\succ_0,\switch(\succ_1))$ but not in $D(\succ_0,\succ_1)$, so these pairs contribute $+1-1=0$ to the difference.
\end{itemize}
\end{itemize}
We proved that the contribution of each pair of alternatives is at least $0$, and the contribution of the pair $\{a,b\}$ is $+1$. Therefore, the difference 
$|D(\succ_1,\succ_0)|-|D(\switch(\succ_1),\succ_0)|$
 is positive and the lemma is proved.
\end{proof}

\newpage
\section{Probability that binomial random variables are equal}
\label{sec:binomial}
This section justifies the following approximation, used in Section \ref{sec:prob}.

Let $B_1,\ldots,B_t$ be i.i.d. random variables distributed binomially with $m$ trials and success-probability $p < 1/2$.
Let $q:=1-p$. 
Define:
\begin{align*}
\peq(m,p,t) := \Pr[B_1 = \cdots = B_t]
\end{align*}
Then, for every $t\geq 1$, when $m$ is sufficiently large, 
\begin{align*}
\peq(m,p,t)
&\approx
{1
\over 
(\sqrt{2\pi p q m})^{t-1}
}
\end{align*}

\begin{proof}
The value of each of the variables $B_i$ can be any integer between 0 and $m$. Therefore we can present $\peq$ as a sum of probabilities of disjoint events:
\begin{align*}
\peq(m,p,t) = 
\sum_{i=0}^m
\Pr[B_1 = \cdots = B_t = i]
\end{align*}
Since the $B_i$ are i.i.d:
\begin{align*}
\peq(m,p,t) &= 
\sum_{i=0}^m
\left(\Pr[B_1 = i]\right)^t
\\
&=
\sum_{i=0}^m
\left({m\choose i} p^i q^{m-i}\right)^t
\end{align*}
where $q:=1-p$.
Using Stirling's approximation, when $m,i,m-i$ are sufficiently large:
\begin{align*}
{m\choose i}&\approx \sqrt{m\over 2 \pi i (m-i)}\cdot {m^m\over i^i (m-i)^{m-i}}
\end{align*}
Substitute this in $\peq$ and approximate the sum by an integral:
\begin{align*}
\peq(m,p,t) &\approx
\int_{y=0}^m
\left(\sqrt{m\over 2 \pi y (m-y)}\cdot {m^m\over y^y (m-y)^{m-y}}\cdot p^y q^{m-y}\right)^t dy
\end{align*}
Substitute $y = m x$ and $dy = m dx$:
\begin{align*}
\peq(m,p,t) &\approx
m\int_{x=0}^1
\left(\sqrt{1 \over 2 \pi m x (1-x)}\cdot {p^{m x} q^{m-mx}\over x^{m x} (1-x)^{m-mx}}\right)^t dx
\\
&=
m^{1-{t/ 2}}\int_{x=0}^1
\left(\sqrt{1 \over 2 \pi x (1-x)}\cdot {p^{m x} q^{m-mx}\over x^{m x} (1-x)^{m-mx}}\right)^t dx
\end{align*}

The integral can be approximated by Laplace's method. Define:
\begin{align*}
h(x) &:=  
\left(\sqrt{1 \over 2 \pi x (1-x)}\right)^t
\\
g(x) &:= 
t \cdot 
\bigg[
x \ln({p\over x})
+
(1-x) \ln({q\over 1-x})
\bigg]
\end{align*}
Then:
\begin{align*}
\peq(m,p,t) &\approx
m^{1-{t / 2}} \int_{x=0}^1
h(x)
e^{m g(x)} dx
\end{align*}
The function $g(x)$ is twice continuously differentiable on $(0,1)$ and has a unique maximum at $x_0={p\over p+q} = p$; the maximum value is $g(x_0) = 0$.
Moreover, $g''(x_0) = -{1 \over p q}< 0$. 
Therefore, by Laplace's method:
\begin{align*}
\peq(m,p,t)
&\approx
m^{1-t/2} \sqrt{
2\pi \over 
- m g''(x_0)
}
\cdot
h(x_0)
\cdot
e^{m g(x_0)}
\end{align*}
where the symbol $\approx$ means that the ratio between the expressions in its two sides goes to $1$ as $m\to\infty$. Substituting the functions $g$ and $h$ gives:
\begin{align*}
\peq(m,p,t)
&\approx
m^{1-t/2} \sqrt{
2\pi \over 
m / (p q)
}
\cdot
\left({1\over 2\pi p q}\right)^{t/2}\cdot e^0
\\
&=
m^{(1-t)/2}
\cdot 
(2\pi p q)^{(1-t)/2}
\\
&=
(2\pi p q m)^{(1-t)/2}
\\
\peq(m,p,t)
&\approx
{1
\over 
(\sqrt{2\pi p q m})^{t-1}
}
\end{align*}
\end{proof}
Note that $\peq(m,p,1) = 1$, which is trivially true, since a single random variable always equals itself.
When $t\geq 2$, $\peq(m,p,1)\to 0$ as $m\to\infty$.
\newpage
\section*{References}
\bibliographystyle{apalike}
\bibliography{consensus}
\end{document}